\documentclass[twocolumn,prx,aps,showpacs]{revtex4}
\usepackage{xspace,amsmath,amsfonts,amsthm,amssymb,amsbsy,graphics,graphicx,bm}
\usepackage{epigraph}
\usepackage[usenames,dvipsnames]{color}
\usepackage{soul} 
\usepackage{epstopdf}
\usepackage{bbding}
\usepackage{pifont}

\usepackage[breaklinks=true]{hyperref}

\hypersetup{
  colorlinks   = true, 
  urlcolor     = blue, 
  linkcolor    = blue, 
  citecolor   = red 
}

\newcommand{\smallfrac}[2]{\mbox{$\frac{#1}{#2}$}}
\newcommand{\half}{\smallfrac{1}{2}}

\newcommand{\bra}[1]{\left\langle{#1}\right|}
\newcommand{\ket}[1]{\left|{#1}\right\rangle}
\newcommand{\op}[2]{\ket{#1}\!\bra{#2}}

\newcommand{\ip}[2]{\left\langle{#1}\right|\left.{#2}\right\rangle}
\newcommand\ipo[3]{\langle#1 |#2| #3 \rangle}
\newcommand{\En}[1]{\mathbb{E}\left [{#1}\right]}

\newcommand{\sE}{\mathcal{E}}

\newcommand{\sF}{\mathcal{F}}
\newcommand{\sG}{\mathcal{G}}
\newcommand{\MSE}{\mbox{MSE}}
\newcommand{\Tr}[1]{\mbox{\rm Tr}\!\left ( {#1}\right )}

\newcommand{\erf}[1]{Eq.~(\ref{#1})}
\newcommand{\frf}[1]{Fig.~\ref{#1}}
\newcommand{\srf}[1]{Sec.~\ref{#1}}
\newcommand\dg{^\dagger}

\setstcolor{Magenta}

\newcommand{\ssQ}{{\scriptscriptstyle{Q}}}
\newcommand{\ssA}{{\scriptscriptstyle{A}}}

\newcommand{\dingcross}{{\mbox{\ding{53}}}}
\newcommand{\dingcrossscript}{{\mbox{\scriptsize\ding{53}}}}
\newcommand\Id{\mathbb{I}}
\newcommand{\ntrue}{\bm{n}_{\rm true}}
\newcommand{\nest}{\bm{n}_{\rm est}}

\newtheorem*{thrm}{Theorem}
\newtheorem*{lem}{Lemma}

\begin{document}
\title{Quantum limits on post-selected, probabilistic quantum metrology}

\author{Joshua Combes}
\author{Christopher Ferrie}
\author{Zhang Jiang}
\author{Carlton M.~Caves}

\affiliation{Center for Quantum Information and Control, University of New Mexico, Albuquerque, New Mexico, 87131-0001}

\begin{abstract}
Probabilistic metrology attempts to improve parameter estimation by occasionally reporting an excellent estimate and the rest of the time either guessing or doing nothing at all. Here we show that probabilistic metrology can never improve quantum limits on estimation of a single parameter, both on average and asymptotically in number of trials, if performance is judged relative to mean-square estimation error.  We extend the result by showing that for a finite number of trials, the probability of obtaining better estimates using probabilistic metrology, as measured by mean-square error, decreases exponentially with the number of trials.  To be tight, the performance bounds we derive require that likelihood functions be approximately normal, which in turn depends on how rapidly specific distributions converge to a normal distribution with number of trials.
\end{abstract}

 \date{\today}

\maketitle

\section{Introduction}\label{sec:intro}
Quantum metrology is the study of how accurately physical quantities can be measured within the structure of quantum theory.  Quantum-limited metrology is important for such objectives as improving time and frequency standards~\cite{HinShePhil13,UdeHolHan02} and detecting gravitational waves~\cite{LIGO}.  The basic scenario of quantum metrology is that of a quantum system whose state is influenced by one or more parameters that are to be determined by probing the system with other physical systems.  The goal is to tailor the measurement made by the probes so as to learn as much as possible about the values of the parameters.  Researchers have been busy devising new ways to estimate parameters~\cite{HigBerBar07} and devising new bounds on how accurately one can estimate a parameter~\cite{BoiFlaCav07,TsaWisCav11,EscMatDav11,KnySmeDur11,HalWis12,Tsang12,GioLloMac12}.

Recently some researchers have proposed ingenious techniques that appear to allow for an improvement in estimation precision beyond the limits quantum mechanics usually imposes.  These techniques go under the names of ``probabilistic metrology"~\cite{Fiurasek06,ChiYanYao2013,Marek2013}, ``metrology with abstention''~\cite{GenRonCal2012,GenRonCal2013a,GenRonCal2013b}, and ``weak-value amplification''~\cite{KofAshNor12,DreMalFil13,JorMarHow13}.  Before examine probabilistic protocols for estimation, we briefly summarize some ways probabilistic protocols are used in quantum information.

Probabilistic protocols have a long history in quantum measurement and information theory and have proven to be very useful in some contexts.  For example, in quantum optics, single-photon states can be probabilistically created by heralding on one photon from a pair created by spontaneous parametric downconversion~\cite{HonMan86}.  The low success probabilities seen in experiments are not fundamental and can be increased by multiplexing~\cite{ColXioRey13}.  Similarly, in the Knill-Laflamme-Milburn quantum-optical computing scheme, the probabilistic gates can have arbitrarily high success rates if special resource states can be prepared offline~\cite{KLM01}.

Using unambiguous state discrimination~(USD)~\cite{Che98,CheBar98}, one can discriminate without error between linearly independent pure states, provided there is the possibility that some of the time one is not required to make a decision among the states.  USD is interesting in its own right, as is probabilistic metrology, but to assess the usefulness of USD for some quantum-information task, one must formulate a performance metric that weighs the trade-off between never making a mistake and sometimes not making a decision.  In metrological contexts, one doesn't have to formulate a performance metric; there typically is a natural  performance metric related to a measure of the accuracy with which the parameters are estimated.  The aim of this paper is to assess probabilistic metrology relative to such a natural performance metric.  Our focus in this paper is not whether protocols for probabilistic quantum metrology are intrinsically interesting, but rather whether they are useful for reaching or beating quantum limits.

The strategy employed by probabilistic metrology schemes is to make a selection measurement that ``concentrates'' information about the parameters into some subset of the measurement outcomes.  Further measurements to determine the parameters, made on these favorable outcomes, provide a refined estimates of the parameters; unfavorable outcomes are discarded with no attempt to gather information from them.  The process of waiting for a favorable outcome is called post-selection\/; it amounts to using the selection measurement to prepare states that, on favorable outcomes, provide high sensitivity to the parameters. In prior work, it is not at all clear if the process of post-selection can aid overall estimation accuracy, primarily because the relative scarcity of the favorable outcomes and/or the failure to garner any information from the discarded outcomes is not fully included in the analysis.

In broad terms, the aim of probabilistic metrology is to use the selection measurement to {\em increase\/} the distinguishability of quantum states.  In this paper, we explore whether this strategy can provide benefits for quantum metrology once the probability of favorable outcomes is properly taken into account.  We focus on estimation of a single parameter~$x$ and show that if the performance of a strategy for estimating $x$ is judged by the mean-square estimation error (MSE), post-selection can never improve on quantum limits for estimating~$x$.  The framework for our analysis is laid out in Secs.~\ref{sec:parest} and~\ref{sec:ancilla}, and the main results are given in Secs.~\ref{sec:results_1} and~\ref{sec:results_2}.

In Sec.~\ref{sec:prior}, we survey previous critical analyses of probabilistic metrology (in Sec.~\ref{subsec:sketchprior}), formulate desiderata for analyzing any protocol for probabilistic metrology (Sec.~\ref{subsec:desiderata}), discuss proposals for metrology using weak-value amplification, to which our results apply, but which has also been analyzed critically elsewhere~\cite{TanYam2013,FerCom2013,Knee2013a,Knee2013b}~(Sec.~\ref{subsec:weakvalueamp}), and consider in some detail the protocol for metrology with abstention formulated in~\cite{GenRonCal2012}~(Sec.~\ref{subsec:abstention}).

\section{Quantum parameter estimation}\label{sec:parest}
We are interested in estimating a parameter $x$ that is impressed on the system state through some general quantum evolution $\rho(x) = \sE_x(\rho)$, where $\sE_x$ is a trace-preserving completely positive map that depends on $x$.  The parameter could correspond, for example, to a unitary phase shift or to the decay constant of an atom.  In order to compare protocols fairly, we need a way to characterize how well an estimation procedure performs.  We compare protocols based on their mean-square estimation error (MSE),
\begin{equation}
\MSE(x_{\rm true})=\mathbb E_{\rm data}\!\left[\bigl(\hat x({\rm data})-x_{\rm true}\bigr)^2\right],
\end{equation}
where $\hat x$ is an estimator for $x_{\rm true}$ and $\mathbb E_{\rm data}[\,\cdot\,]$ denotes an expectation over data.

For quantum estimation problems there is a strict lower bound on the MSE obtained from any unbiased estimator applied to data collected from any quantum measurement.  The bound is expressed in terms of the quantum Fisher information $I_\rho(x)$ associated with the state $\rho(x)$ that encodes the parameter.  For measurements on $N$ copies of the system---we call these $N$ trials---the bound, called the {\em quantum Cram\'er-Rao bound\/} (QCRB), is expressed as
\begin{equation}\label{eq:Fisher-bound}
{\rm MSE}(x) \geq \frac{1}{N I_\rho(x)}.
\end{equation}
The quantum Fisher information is defined as~\cite{Helstrom1976Quantum}
\begin{equation}\label{eq:Irhox}
I_\rho(x) = \Tr{\rho(x)L(x)^2},
\end{equation}
where $L(x)$ is the \emph{symmetric logarithmic derivative\/} (SLD), implicitly defined by the equation
\begin{equation}\label{SLD}
\frac{\partial}{\partial x}\rho(x) = \frac12 \big[ L(x)\rho(x) + \rho(x)L(x) \big].
\end{equation}

It is known that there is an optimal quantum measurement whose \emph{classical\/} Fisher information, calculated from the probabilities of measurement outcomes, achieves the \emph{quantum\/} Fisher information~(\ref{eq:Irhox})~\cite{BraunsteinCaves1994}; letting $\{\Upsilon_k\}$  be any positive-operator-valued measure (POVM), with outcome probabilities $p(k|x)=\Tr{\rho(x)\Upsilon_k}$,
this can be written as
\begin{equation}\label{classical to quantum fisher}
I_\rho(x) = \max_{\{\Upsilon_k\}}\,I_{\rm cl}\big[p(k|x)\big],
\end{equation}
where the classical Fisher information of the outcome probabilities is
\begin{equation}
I_{\rm cl}\big[p(k|x)\big]=\sum_k p(k|x)\!\left(\frac{\partial\ln p(k|x)}{\partial x}\right)^2.
\end{equation}
Equation~(\ref{classical to quantum fisher}) is an alternative, operational definition of the quantum Fisher information.  An additional important fact is that the maximum likelihood estimator generally achieves the classical Fisher bound asymptotically in $N$, and this, together with \erf{classical to quantum fisher}, means that the QCRB does indeed express the quantum limit on achievable MSE.  To achieve the QCRB often requires prior information about the parameter and, in practice, an adaptive implementation of the optimal measurement.

Before moving on to our detailed analysis, we draw attention to one important point.  The inequalities we derive are tight only when the relevant likelihood function is approximately normal.  Thus the rate of convergence with $N$ of a particular likelihood function to the normal distribution determines when the following bounds are tight.  For measurements with Gaussian statistics, the inequalities are saturated immediately, for one trial. Interestingly, distributions that are quite different from the normal distribution rapidly become approximately normal.  For example (see Chap.~2 of~\cite{BoxHunHun05}), ``a rough rule'' is that for $N>5$, the normal approximation to a binomial distribution is good if $|(\sqrt{(1-p)/p}-\sqrt{p/(1-p)})/\sqrt N|<0.3$, where $p$ is the Bernoulli or binomial success probability.  In practice, provided $p$ is not too close to 0 or 1, this means that $20\alt N\alt 100$ is sufficient to approximate a normal distribution.  As quantum statistics are often Bernoulli or binomial, we can expect our results to hold very closely in as few as $20$--$100$ trials.  In typical quantum-metrology applications~\cite{SerChaCom11}, the likelihood functions are highly nonGaussian and oscillatory, yet the convergence to a normal distribution still occurs approximately after as little as $15$ trials~(see Sec.~6 of~\cite{FerGraCor13}).  Braunstein~\cite{Braun92} has explored the question of approach to Gaussian statistics in the context of the maximum likelihood estimation that can achieve the Cram\'er-Rao bound.

\section{Probabilistic quantum metrology and ancilla model}\label{sec:ancilla}
Rather than using the optimal POVM which solves Eq.~(\ref{classical to quantum fisher}), the idea of probabilistic metrology is to ``encode'' the information about the unknown parameter ``in a more efficient way''~\cite{GenRonCal2012,GenRonCal2013a,GenRonCal2013b,ChiYanYao2013,Marek2013,DreMalFil13,JorMarHow13}.  Formally, one makes a selection measurement whose outcomes are divided into the set $\checkmark$ of favorable outcomes, which ``concentrate'' the information about the parameter, and the complementary set $\dingcross$ of unfavorable outcomes, which are discarded.  When a favorable outcome is obtained, a second measurement is performed on the post-selected state to extract information about $x$.

To investigate these ideas, we compare the quantum Fisher information before and after the selection measurement.  This is different from the analysis performed in~\cite{TanYam2013,FerCom2013}, where there were additional assumptions about how the parameter was encoded in the state $\rho(x)$ and about what types of selection measurement could be performed.  Here we use the most general forms allowed by quantum mechanics.

The quantum state that encodes the classical parameter $x$ is denoted by $\rho_{\ssQ}(x) = \sE_x(\rho)$, where we now label the system with $Q$ to distinguish it from the ancilla we introduce shortly.  The system can consist of more than one part: in weak-value amplification, for example, the system $Q$ is divided into two parts, $R$ and $S$, the parameter is the strength of an interaction between $R$ and $S$, and the selection measurement is performed on $R$ alone~\cite{FerCom2013}. Notice also that since we allow the encoded state $\rho_\ssQ(x)$ to be mixed, our analysis covers the case of technical noise that is imposed on the system as the parameter is encoded in the system.  We can simply regard the operation $\sE_x$ as incorporating a description of such technical~noise.

The selection measurement is described by quantum operations, one for each outcome~$\alpha$,
\begin{align}\label{eq:uncond_op}
\sF_\alpha\!\left[\rho_{\ssQ}(x)\right] =
\sum_{j=0}^{J_\alpha-1}M_{\alpha,j} \,\rho_{\ssQ}(x)\, M_{\alpha,j}\dg,
\end{align}
where the operators $M_{\alpha,j}$ are Kraus operators.  The POVM element for outcome~$\alpha$ and the completeness relation satisfied by the POVM elements are
\begin{equation}
E_\alpha= \sum_{j=0}^{J_\alpha-1}M_{\alpha,j}\dg M_{\alpha,j},
\qquad
\Id_\ssQ=\sum_{\alpha=1}^J E_\alpha.
\end{equation}
The subscript $j$ allows for the possibility that the quantum operations involve course graining over measurement results we don't have access to.  Our subsequent analysis requires us to be clear about the values assumed by $\alpha$ and $j$: the outcomes $\alpha$ are labeled by positive integers, $1,\ldots,J$, and the index $j$, when associated with outcome $\alpha$, takes on values $j=0,1,\ldots,J_\alpha-1$.

The post-measurement state, post-selected on outcome $\alpha$, is
\begin{align}\label{eq:cond_op}
\sigma_{\ssQ|\alpha}(x) = \frac{\sF_\alpha[\rho_{\ssQ}(x)]}{p(\alpha|x)},
\end{align}
where $p(\alpha|x) =\Tr{\sF_\alpha[\rho_{\ssQ}(x)]} = \Tr{E_\alpha \rho_{\ssQ}(x)}$ is the conditional probability of obtaining outcome $\alpha$ given the state $\rho_{\ssQ}(x)$.

The idea behind probabilistic quantum metrology is that the states for favorable outcomes ($\alpha\in\checkmark$) have higher Fisher information than $\rho_\ssQ(x)$.  The QCRB arises, however, from the quantum Fisher information of an individual quantum state and not directly from an average of the quantum Fisher informations for states occurring with various probabilities.  Thus, to formalize the idea of probabilistic metrology for analysis using the QCRB, we need to formulate it in terms of a single quantum state.  For this purpose, we employ an ancilla $A$ that records and stores the outcomes of the selection measurement.  The joint system-ancilla state we are shooting for is
\begin{align}\label{sigmanocheck}
\sigma_{\ssQ\ssA}(x)=\sum_{\alpha=1}^J p(\alpha|x)\sigma_{\ssQ|\alpha}(x)\otimes\op{f_\alpha}{f_\alpha},
\end{align}
where the states $\ket{f_\alpha}$ are orthonormal ancilla states.  In the state~(\ref{sigmanocheck}) the ancilla stores a record of the selection-measurement outcomes; the outcomes are correlated with the post-selected system states $\sigma_{\ssQ|\alpha}(x)$, which occur with probability $p(\alpha|x)$.

We now show how to get to the state~(\ref{sigmanocheck}) and to successor states that are relevant for probabilistic metrology by physical processes; this demonstration illuminates how information is discarded at various points in these processes. To reiterate where we are headed, we are going to show in Sec.~\ref{sec:results_1} how getting to the state~(\ref{sigmanocheck}) and to the successor states envisioned by probabilistic metrology decreases the Fisher information available for parameter estimation.

We first invoke the Kraus representation theorem~\cite{Kraus83,MikeandIke}, which tells us that given the complete set of Kraus operators, $\{M_{\alpha,j}\}$, there exists an ancilla with initial pure state $\rho_\ssA=\op{\psi}{\psi}$ and a joint unitary operator $U$ such that $M_{\alpha,j}= \ipo{f_{\alpha,j}}{U}{\psi}$, where the states $\ket{f_{\alpha,j}}$ make up an orthonormal basis for the ancilla.  The evolution under $U$ can be written as
\begin{equation}\label{eq:firstjoint}
U\rho_\ssQ(x)\otimes\rho_\ssA U^\dagger
=\sum_{\alpha,j;\beta,k}M_{\alpha,j}\rho_\ssQ(x)M\dg_{\beta,k}\otimes\op{f_{\alpha,j}}{f_{\beta,k}}.
\end{equation}
This joint unitary evolution does not store the measurement outcome in the ancilla.  Indeed, since the unitary $U$ can be reversed, the state~(\ref{eq:firstjoint}) has the same quantum Fisher information as $\rho_\ssQ(x)$.  To record and store the outcome in the ancilla requires some decoherence of the ancilla, which can be thought of as a measurement on the ancilla.

Na\"{\i}vely one might expect that performing the measurement specified by Kraus operators $\Omega_\alpha = \Id_\ssQ \otimes \sum_j \op{f_{\alpha,j}}{f_{\alpha,j}}$ would result in the state~(\ref{sigmanocheck}), but in fact, it gives the state
\begin{align}
\sum_\alpha\Omega_\alpha&U\rho_\ssQ(x)\otimes\rho_\ssA U^\dagger\Omega_\alpha\nonumber\\
&=\sum_{\alpha,j,k}M_{\alpha,j}\rho_\ssQ(x)M\dg_{\alpha,k}\otimes\op{f_{\alpha,j}}{f_{\alpha,k}}.
\end{align}
This measurement, though it removes the coherence between different outcome subspaces in the ancilla, leaves the coherence within each outcome subspace and thus does not reproduce the state~(\ref{sigmanocheck}).

To get to the state~(\ref{sigmanocheck}), we need to do a measurement on the ancilla in the basis $\{\ket{f_{\alpha,j}}\}$ and then keep only the outcome $\alpha$; erasing the suboutcome~$j$ can be done by a post-measurement unitary on the ancilla that, in each outcome subspace $\alpha$, leaves the ancilla in a particular state, which we take to be $\ket{f_{\alpha,0}}$.  The desired ancilla measurement is thus described by the quantum operation
\begin{align}\label{eq:sG}
\sG=\sum_{\alpha=1}^J\sum_{j=0}^{J_\alpha-1}K_{\alpha,j}  \odot K\dg_{\alpha,j},
\end{align}
where the Kraus operators are given by
\begin{equation}
K_{\alpha,j}= \Id_\ssQ\otimes\op{f_{\alpha,0}}{f_{\alpha,j}}.
\end{equation}
In \erf{eq:sG}, the symbol $\odot$ is a place holder for the operator the operation acts on.   Applying $\sG$ to the state~(\ref{eq:firstjoint}) gives the desired state~(\ref{sigmanocheck}),
\begin{align}\label{sigmanocheck2}
\sG\big[U\rho_\ssQ(x)\otimes\rho_\ssA U^\dagger\big]
&=\sum_{\alpha=1}^J\sF_\alpha[\rho_\ssQ(x)]\otimes\op{f_\alpha}{f_\alpha}\nonumber\\
&=\sigma_{\ssQ\ssA}(x),
\end{align}
where we identify $\ket{f_\alpha}=\ket{f_{\alpha,0}}$.  The quantum operation $\sG$ is a decoherence process that stores the selection-measurement outcomes in the ancilla, correlated with the post-selected system states $\sigma_{\ssQ|\alpha}(x)$.

Before proceeding further, a remark is in order.  Had we restricted our analysis to selection-measurement quantum operations each of which has a single Kraus operator, i.e., replaced \erf{eq:uncond_op} with $\sF_\alpha\left[\rho_{\ssQ}(x)\right] = M_{\alpha} \,\rho_{\ssQ}(x)\, M_{\alpha}\dg$, the analysis to this point would be considerably simplified at the cost of less generality.

Having gotten to the state~(\ref{sigmanocheck}), we now imagine a further conditional decoherence that for the unfavorable outcomes, damps the system to a state $\sigma_{\ssQ|\dingcrossscript}=\op{\phi_\dingcrossscript}{\phi_\dingcrossscript}=\sigma_{\ssQ|0}$, which contains no information about $x$ (zero Fisher information), and the ancilla to a state $\ket{f_\dingcrossscript}=\ket{f_0}$, which can be taken to be the state $\ket{f_\alpha}$ for a particular unfavorable outcome. The result is the joint state
\begin{align}\label{sigmacheck}
\sigma_{\ssQ\ssA,\checkmark}(x)
&=\sum_{\alpha\in\checkmark}p(\alpha|x)\sigma_{\ssQ|\alpha}(x)\otimes\op{f_\alpha}{f_\alpha}\nonumber\\
&\phantom{=\sum_{\alpha\in\checkmark}}
+p(\dingcross|x)\sigma_{\ssQ|\dingcrossscript}\otimes\op{f_\dingcrossscript}{f_\dingcrossscript}\\
&=\sum_{\alpha\in\{0,\checkmark\}} p_\checkmark(\alpha|x)\sigma_{\ssQ|\alpha}(x)\otimes\op{f_\alpha}{f_\alpha}.\nonumber
\end{align}
Here
\begin{equation}
p(\dingcross|x)=\sum_{\alpha\in\dingcrossscript}p(\alpha|x)
\end{equation}
is the total probability of the unfavorable outcomes.  To simplify our expressions we introduce, in the second form of \erf{sigmacheck}, the conditional distribution $p_\checkmark(\alpha|x)$.  This distribution is defined on the lumped unfavorable outcomes, which are labeled by $\alpha=\dingcross$ or $\alpha=0$ (we find both these labels to be useful), and the favorable outcomes by
\begin{equation}
p_\checkmark(\alpha|x)=
\begin{cases}
p(\dingcross|x),&\alpha=0,\\
p(\alpha|x),&\alpha\in\checkmark.
\end{cases}
\end{equation}

The state~(\ref{sigmacheck}) encodes the parameter $x$ in a way that is envisioned by probabilistic quantum metrology: the ancilla records the outcome of the selection measurement; the post-selected system states $\sigma_{\ssQ|\alpha}(x)$ for the favorable outcomes occur with the right probabilities $p(\alpha|x)$; and the unfavorable outcomes are lumped together and associated with a state $\sigma_{\ssQ|0}$ that has no information about~$x$.  All information about the parameter has clearly been removed from the unfavorable outcomes, but the state~(\ref{sigmacheck}) still allows a guess for the parameter when an unfavorable outcome occurs.

We consider one other state, which is perhaps the best expression of the strategy of probabilistic metrology.  This state arises from looking at the ancilla and, if the outcome is favorable, handing the resulting state to a party who performs the rest of the probabilistic-metrology protocol.  The probability of the hand-off is the total probability of the favorable outcomes,
\begin{align}
p(\checkmark|x)=\sum_{\alpha\in\checkmark}p(\alpha|x)=1-p(\dingcross|x),
\end{align}
and the state handed to the other party is
\begin{align}\label{sigmarealcheck}
\sigma_{\ssQ\ssA|\checkmark}(x)
&=\frac{\Pi_\checkmark\sigma_{\ssQ\ssA,\checkmark}(x)\Pi_\checkmark}{p(\checkmark|x)}\nonumber\\
&=\sum_{\alpha\in\checkmark}\frac{p_\checkmark(\alpha|x)}{p(\checkmark|x)}
\sigma_{\ssQ|\alpha}(x)\otimes\op{f_\alpha}{f_\alpha}.
\end{align}
Here $\Pi_\checkmark=\sum_{\alpha\in\checkmark}\op{f_\alpha}{f_\alpha}$ projects the ancilla onto the favorable outcomes.  With the state~(\ref{sigmarealcheck}), an estimate of the parameter is made only on the favorable outcomes.

\begin{figure*}[ht]
\includegraphics[width=0.95\linewidth]{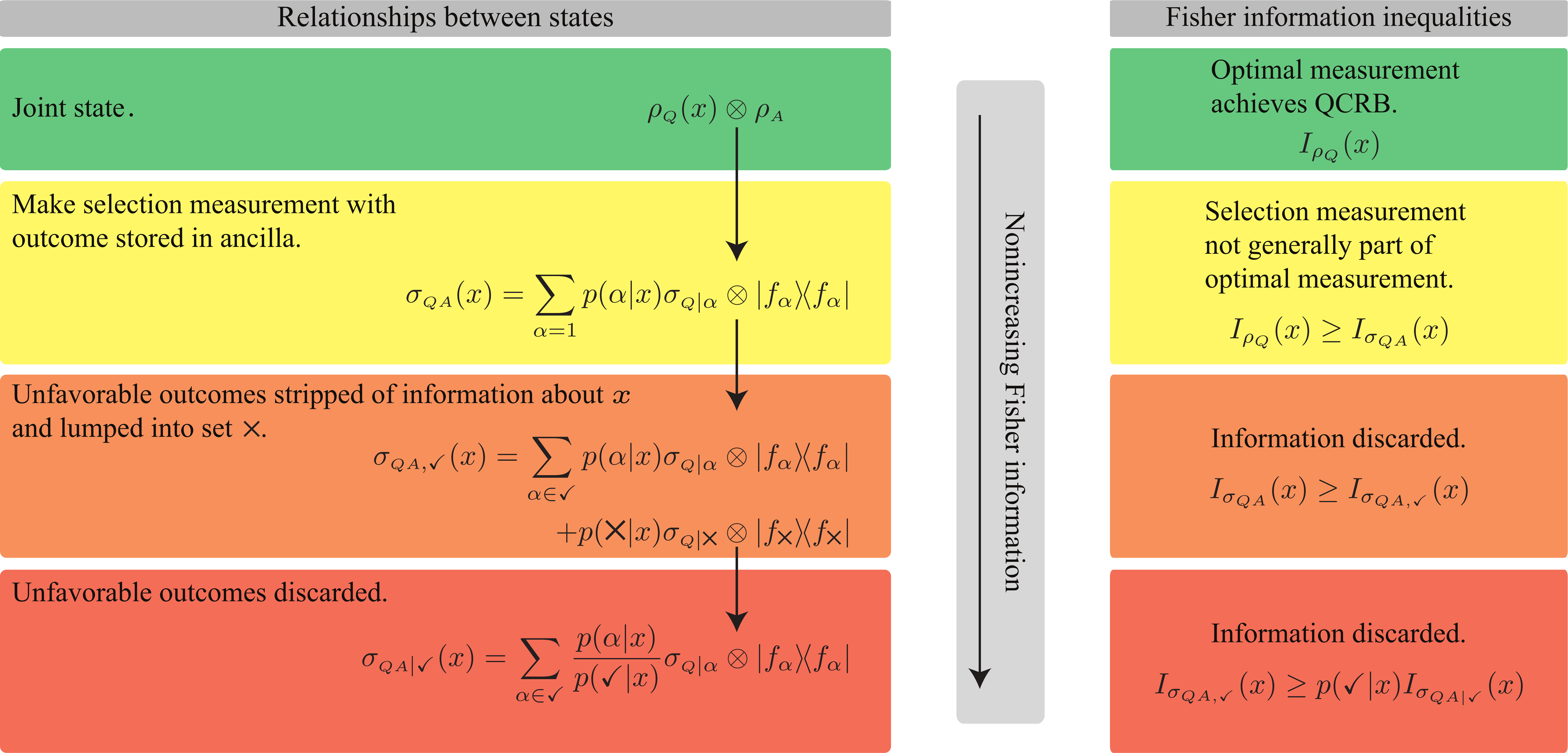}
\caption{\label{fig}
Summary of the relationships among the states in our analysis: the diagram should be read as a threat advisory, where the danger is decreased Fisher information.  Our main results are that the Fisher information cannot increase as one goes from the top of the diagram to the bottom; the decrease in going to the bottom state is statistical in a sense analyzed further in \srf{sec:results_2}.  The boxes in the right column give the inequalities between the Fisher informations and cite which of the principles enunciated in \srf{sec:results_1} encapsulates why the Fisher information decreases.
}
\end{figure*}

Our main results are concerned with the Fisher information for the joint states~(\ref{sigmanocheck}), (\ref{sigmacheck}), and~(\ref{sigmarealcheck}).  Notice that information about $x$ is thrown away in going from
the unconditional state $\sigma_{\ssQ\ssA}(x)$ to $\sigma_{\ssQ\ssA,\checkmark}(x)$ and again in going from $\sigma_{\ssQ\ssA,\checkmark}(x)$ to $\sigma_{\ssQ\ssA|\checkmark}(x)$.  These states and their relations are summarized in Fig.~\ref{fig}.

\section{Fisher-information inequalities}\label{sec:results_1}
To highlight our main results and to isolate the main technical manipulations in a proof, we style them as a Lemma and a Theorem.  The Lemma is concerned with the Fisher information for the states~(\ref{sigmanocheck}), (\ref{sigmacheck}), and~(\ref{sigmarealcheck}). We remind the reader that the QCRBs that come from the quantum Fisher informations below can generally only be achieved asymptotically in the number of trials, i.e., $N\gg1$.

\begin{lem}\label{lem}
The quantum Fisher information of the state~(\ref{sigmacheck}) is
\begin{align}\label{eq:Fisher1}
I_{\sigma_{QA,\checkmark}}\!(x)
=I_{\rm cl}[p_\checkmark(\alpha|x)]+\sum_{\alpha\in\checkmark}p(\alpha|x)I_{\sigma_{\ssQ|\alpha}}\!(x);
\end{align}
i.e., it is the sum of the classical Fisher information of the distribution $p_\checkmark(\alpha|x)$ and the average quantum Fisher information of the favorable-outcome states.  The quantum Fisher information of the state~(\ref{sigmarealcheck}) is related to the Fisher information~(\ref{eq:Fisher1}) by
\begin{equation}\label{eq:Fisher2}
I_{\sigma_{QA,\checkmark}}\!(x)
=p(\checkmark|x)I_{\sigma_{QA|\checkmark}}\!(x)+I_{\rm cl}[p(\checkmark|x),p(\dingcross|x)],
\end{equation}
where the final term is the classical Fisher information for the binary distribution of favorable vs.\ unfavorable outcomes.  When the favorable set includes all outcomes, Eq.~(\ref{eq:Fisher1}) becomes the Fisher information
of the unconditional state~(\ref{sigmanocheck}):
\begin{align}\label{eq:Fisher3}
I_{\sigma_{QA}}\!(x)
=I_{\rm cl}[p(\alpha|x)]+\sum_{\alpha=1}^J p(\alpha|x)I_{\sigma_{\ssQ|\alpha}}\!(x).
\end{align}
\end{lem}

\begin{proof}
The proof is a mainly a straightforward derivation of the SLD for $\sigma_{\ssQ\ssA,\checkmark}(x)$.  We first take the derivative of $\sigma_{\ssQ\ssA,\checkmark}(x)$ with respect to the parameter $x$:
\begin{align}\label{generator}
\partial_x \sigma_{\ssQ\ssA,\checkmark}(x)=
&\sum_{\alpha\in\{0,\checkmark\}}
\big[\sigma_{\ssQ|\alpha}(x)\partial_x p_\checkmark(\alpha|x)\nonumber\\
&\quad+p_\checkmark(\alpha|x)\partial_x\sigma_{\ssQ|\alpha}(x)\big]\otimes\op{f_\alpha}{f_\alpha}.
\end{align}
The derivative $\partial_x\sigma_{\ssQ|\alpha}(x)$ defines the SLD of the states after the selection measurement via
\begin{align}
\partial_x\sigma_{\ssQ|\alpha}(x)=
\half[L_\alpha(x)\sigma_{\ssQ|\alpha}(x)+\sigma_{\ssQ|\alpha}(x)L_\alpha(x)],
\end{align}
and the derivative of the outcome probabilities can be written in terms of the usual classical logarithmic derivative,
$\partial_x p_\checkmark(\alpha|x)=p_\checkmark(\alpha|x)\partial_x\ln p_\checkmark(\alpha|x)$.  We combine the two derivatives in \erf{generator} into the operator
\begin{align}
L_\checkmark(x)=\sum_{\alpha\in\{0,\checkmark\}}[\partial_x\ln p_\checkmark(\alpha|x)+L_\alpha(x)]\otimes\op{f_\alpha}{f_\alpha}.
\end{align}
After a little algebra, one finds that
\begin{align}
\partial_x \sigma_{\ssQ\ssA,\checkmark}(x)
=\half\big[L_\checkmark(x)\sigma_{\ssQ\ssA,\checkmark}(x)+\sigma_{\ssQ\ssA,\checkmark}(x)L_\checkmark(x)\big],
\end{align}
which makes $L_\checkmark(x)$ the SLD of $\sigma_{\ssQ\ssA,\checkmark}(x)$.  That the ancilla states $\ket{f_\alpha}$ are orthogonal is the crucial part of the algebra; this is the mathematical expression of the fact that the ancilla stores a record of the selection-measurement outcomes.

The SLD in hand, we can compute the quantum Fisher information $I_{\sigma_{QA,\checkmark}}\!(x)=\Tr{\sigma_{\ssQ\ssA,\checkmark}(x)L_\checkmark^2(x)}$, using
\begin{align}
\sigma_{\ssQ\ssA,\checkmark}(x)&L_\checkmark^2(x)
=\sum_{\alpha\in\{0,\checkmark\}}
p_\checkmark(\alpha|x)\sigma_{\ssQ|\alpha}(x)\nonumber\\
&\times\bigl[\partial_x\ln p_\checkmark(\alpha|x)+L_\alpha(x)\bigr]^2\otimes\op{f_\alpha}{f_\alpha}.
\label{eq:lem}
\end{align}
Here we again use the orthogonality of the states $\ket{f_\alpha}$.  In evaluating the trace to find the Fisher information, the cross terms that come from the square in \erf{eq:lem} vanish because $\Tr{L_\alpha(x)\sigma_{\ssQ|\alpha}(x)}=\Tr{\partial_x\sigma_{\ssQ|\alpha}(x)}=\partial_x\Tr{\sigma_{\ssQ|\alpha}(x)}=0$.  The result is that the quantum Fisher information of the state $\sigma_{\ssQ\ssA,\checkmark}(x)$ is that given in \erf{eq:Fisher1}.  Notice that $L_0=0$ because the state $\sigma_{\ssQ|0}$ is independent of $x$, so $I_{\sigma_{\ssQ|0}}=0$ does not contribute to \erf{eq:Fisher1}.

The quantum Fisher information of $\sigma_{\ssQ\ssA|\checkmark}$ follows from an identical derivation,
\begin{align}\label{eq:Fisher4}
I_{\sigma_{QA|\checkmark}}\!(x)
=I_{\rm cl}[q_\checkmark(\alpha|x)]+\sum_{\alpha\in\checkmark}q_\checkmark(\alpha|x)I_{\sigma_{\ssQ|\alpha}}\!(x).
\end{align}
where $q_\checkmark(\alpha|x)=p(\alpha|x)/p(\checkmark|x)$ is the renormalized probability of the favorable outcomes.  Straightforward manipulation of $I_{\rm cl}[q_\checkmark(\alpha|x)]$ leads to \erf{eq:Fisher2}.
\end{proof}

We are now ready to state our main result, which generalizes inequalities derived by Tanaka and Yamamoto~\cite{TanYam2013} and Ferrie and Combes~\cite{FerCom2013}.

\begin{thrm}
The quantum Fisher informations of the states introduced above satisfy
\begin{align}
I_{\rho_Q}(x)
\ge I_{\sigma_{QA}}(x)
\ge I_{\sigma_{QA,\checkmark}}(x)
\ge p(\checkmark|x)I_{\sigma_{QA|\checkmark}}(x).
\label{eq:theorem}
\end{align}
When the favorable set $\checkmark$ contains only a single outcome~$\alpha$, the final entry in the chain becomes $p(\alpha|x)I_{\sigma_{Q|\alpha}}(x)$.
\end{thrm}

\begin{proof}
The first inequality follows from \erf{classical to quantum fisher}: either the selection measurement followed by optimal extraction of information about $x$ from $\sigma_{\ssQ\ssA}(x)$ is optimal, or (more likely) it is suboptimal; either way, the inequality holds.  The second inequality follows from a similar optimality argument or directly from comparing Eqs.~(\ref{eq:Fisher1}) and~(\ref{eq:Fisher3}), with the additional fact that lumping classical alternatives together, as in lumping the unfavorable outcomes together, cannot increase the classical Fisher information~\cite{note}.  The third inequality is an immediate consequence of \erf{eq:Fisher2}, and the final sentence is confirmed by \erf{eq:Fisher4}.
\end{proof}

The chain of inequalities in the Theorem, which is summarized in \frf{fig}, shows that probabilistic metrology cannot beat fundamental metrological quantum limits.  The first inequality says that unless the selection measurement and subsequent measurements on the post-selected states are optimal, they cannot do as well as the optimal measurement.  The second inequality says that discarding the information in the post-selected states for unfavorable outcomes cannot improve the quantum limit on estimating~$x$.  The third inequality says that discarding entirely the post-selected states for the unfavorable outcomes cannot improve quantum limits, although the presence of the success probability $p(\checkmark|x)$ requires further discussion, which we give in \srf{sec:results_2}.

Much ink is spilt here in formulating and stating our results precisely, but the results enshrined in the Theorem arise from two principles, which we hold to be self-evident:
\begin{enumerate}
\item[P1:] A suboptimal strategy cannot achieve optimal performance (this is the message of the first \hbox{inequality});
\item[P2:] Information cannot be increased by throwing some of it away (this is the essence of the second and third inequalities).
\end{enumerate}
These two principles should inform thinking about metrology even before things are spelled out precisely. Much has been made of the possibility that one might trade optimality for other practical advantages~\cite{JorMarHow13}.  Indeed, there might be practical advantages to discarding data (such as reduced data processing time), but generally data should be discarded only when it contains no useful information.

The last inequality in \erf{eq:theorem} bears further consideration because the state $\sigma_{\ssQ\ssA|\checkmark}(x)$ is closest to the idea behind probabilistic metrology and because of the presence of the probability $p(\checkmark|x)$ for the favorable outcomes in the inequality.  The Fisher information of the post-selected state $\sigma_{\ssQ\ssA|\checkmark}(x)$ can be larger than the Fisher information of $\rho_\ssQ(x)$---it is such favorable outcomes that one hopes to exploit in probabilistic metrology---but the success probability $p(\checkmark|x)$ says that this strategy can only work in the sense of a bet.  Next we turn our attention to the question of how often this bet can pay off.

\section{Statistical inequalities for gambling on favorable outcomes}\label{sec:results_2}
We can analyze the scenario above formally by imagining a sequence of $N$ independent trials, with $N_\checkmark$ trials having a favorable outcome and, for these favorable trials, a subsequent use of the state~$\sigma_{\ssQ\ssA|\checkmark}(x)$ to estimate~$x$.  The QCRB can be written as
\begin{equation}\label{eq:Fisherbound2}
{\rm MSE}(x) \geq \frac{1}{N_\checkmark I_{\sigma_{QA|\checkmark}}(x)},
\end{equation}
where $N_\checkmark$ is a random variable.  Since $\En{N_\checkmark}=Np(\checkmark|x)$, the last inequality in \erf{eq:theorem} enforces the bound~(\ref{eq:Fisherbound2}) on average or, more precisely, asymptotically for large~$N$. Mightn't it be possible, however, to beat the quantum limit in a finite number of trials that happen to have a large number of favorable outcomes?  Though this can happen, its likelihood can be bounded using standard statistical tools.

Following Ferrie and Combes~\cite{FerCom2013}, what we do is bound the probability that $N_\checkmark I_{\sigma_{QA|\checkmark}}$ exceeds $NI_{\rho_Q}$ or, equivalently, that
\begin{equation}
N_\checkmark\ge \frac{NI_\rho}{I_\sigma}
=\En{N_\checkmark}\frac{I_\rho}{p(\checkmark)I_\sigma}.
\end{equation}
To reduce clutter here and and in the remainder of this section, we omit reference to~$x$, and we use the abbreviations
\begin{equation}
I_\rho=I_{\rho_Q}(x)\quad\mbox{and}\quad I_\sigma=I_{\sigma_{QA|\checkmark}}(x),
\end{equation}
since $\rho_\ssQ$ and $\sigma_{\ssQ\ssA|\checkmark}$ are the only two states involved in the discussion.  We use the \emph{Chernoff bound\/}~\cite{chernoff1952,HoeffdingChernoff1963}, which bounds the probability that a sum, $X$, of random variables, each lying between 0 and 1, is greater than its mean $\mu$ by a factor $\delta\ge1$:
\begin{equation}\label{eq:Cbound1}
\Pr\bigl[X\ge\delta\mu\bigr]\le e^{-\mu(\delta-1)^2/(\delta+1)}.
\end{equation}
For the case at hand, it follows that
\begin{align}
&\Pr\!\left[N_\checkmark I_\sigma\ge NI_\rho\right]\nonumber\\
&\quad=\Pr\!\left[N_\checkmark\ge\frac{NI_\rho}{I_\sigma}\right]
\le\exp\!\left(\!-Np(\checkmark)\frac{(\delta-1)^2}{\delta+1}\right),
\label{eq:bound1}
\end{align}
where $\delta=I_\rho/p(\checkmark)I_\sigma$.  The probability of gaining an advantage from probabilistic metrology is thus exponentially suppressed in the number of trials.

Notice, however, that when $\delta\gg1$, a situation that could easily be encountered and in which we would expect probabilistic metrology to perform poorly, the bound~(\ref{eq:bound1}) becomes
$\Pr\!\left[N_\checkmark I_\sigma\ge NI_\rho\right]
\le\exp(-NI_\rho/I_\sigma)$, which is not at all small when $N\alt I_\rho/I_\sigma$.  Since probabilistic metrology aims to have $I_\sigma\gg I_\rho$, this bound suggests that a probabilistic-metrology protocol might have a high probability of exceeding the QCRB for small numbers of trials.  Our intuition, stemming from the notion that $\delta\gg1$ says that the favorable outcomes contain little of the information in $\rho_\ssQ(x)$, suggests that the bound~(\ref{eq:bound1}) is not very good in these circumstances, and that turns out to be the case.

The task is thus to improve the bound~(\ref{eq:bound1}) to match our intuition, and indeed, the Chernoff bound~(\ref{eq:Cbound1}) is derived from an approximation that works best when $\delta$ is near 1.  The tighter bound, from which \erf{eq:Cbound1} is derived, is
\begin{align}
\Pr\bigl[X\ge\delta\mu\bigr]\le e^{-\mu(1-\delta+\delta\ln\delta)}
=e^{-\mu}(e/\delta)^{\mu\delta}.
\end{align}
Using this bound, we find
\begin{align}
\Pr\!\left[N_\checkmark I_\sigma\ge NI_\rho\right]
\le e^{-Np(\checkmark)}\!\left(\frac{e\,p(\checkmark)I_\sigma}{I_\rho}\right)^{NI_\rho/I_\sigma}.
\end{align}
This bound takes care of the situation described above.  The term in large parentheses, $e/\delta$, is small when $\delta$ is large, and this term is taken to a power that is linear in $N$ and large for all values of $N$ when $I_\rho/I_\sigma\gg1$.

The considerations in this section prompt us to formulate a statistical version of our second principle:
\begin{enumerate}
\item[P2$^\prime$:] Attempts to increase information statistically by discarding information probabilistically are bad bets.
\end{enumerate}
We do, however, caution the reader that since the QCRB can only be achieved for large~$N$, it could be that the favorable-outcome state $\sigma_{\ssQ\ssA|\checkmark}(x)$ converges to its QCRB more rapidly than does the initial state $\rho_\ssQ(x)$.  In this situation, there could be an advantage to post-selection for finite~$N$; establishing such an advantage would require a detailed, case-specific analysis of convergence to the respective QCRBs.

\section{Relationship to prior work}\label{sec:prior}

\subsection{Sketch of prior Fisher-information analyses}\label{subsec:sketchprior}
The results presented above owe much to prior critical analyses of probabilistic metrology in the literature.  Here we sketch some results of other researchers and put their results in context by describing the inequalities they proved in our notation.

Inspired by the results of Knee {\it et al.}~\cite{Knee2013a} a series of papers~\cite{Knee2013a, Knee2013b,TanYam2013,FerCom2013,PanJiaCom13,ZhaDatWam13} have shown that probabilistic metrology in the context of weak-value amplification is not a statistically useful way to design an experiment and then process the results.  Knee {\it et al.}~\cite{Knee2013a} showed that, for a particular estimation problem involving initial pure product states of two qubits and two-outcome projective measurements, the quantum Fisher information obeys the inequality  $I_{\psi_Q}(x)\ge p(\checkmark|x)I_{\psi_{QA|\checkmark}}(x)$, where $\psi$ denotes pure states and where there is only a single outcome in the favorable outcome set.  Tanaka and Yamamoto~\cite{TanYam2013} proved this inequality for any pair of quantum systems that begin in a pure product state and interact via any interaction Hamiltonian.  Ferrie and Combes~\cite{FerCom2013} generalized the inequality of Tanaka and Yamamoto to a double inequality that includes the mixed state $\sigma_{\ssQ\ssA}(x)$, i.e., $I_{\psi_Q}(x)\ge I_{ \sigma_{QA}}(x)\ge p(\checkmark|x)I_{\psi_{QA|\checkmark}}(x)$; in this inequality, the state $\sigma_{\ssQ\ssA}(x)$ is the ensemble of post-selected pure states, and the second inequality explicitly includes contributions from the classical Fisher information of the selection-measurement probabilities.  Ferrie and Combes also considered the presence of arbitrary Gaussian technical noise on the input state, thus demonstrating that the derived inequalities are true even in the presence of such technical noise.   They also used the Chernoff bound, in a special case of the analysis in Sec.~\ref{sec:results_2}, to find the consequences of the Fisher-information inequalities for a finite number of trials.  Recently, Zhang, Datta and Walmsley~\cite{ZhaDatWam13} have independently proved a special case of the Ferrie-Combes inequality and illustrated it with several examples, and Knee and Gauger~\cite{Knee2013b} have shown, using a Fisher-information analysis, that if there is technical noise on the detector, weak-value amplification offers no advantage for overcoming such technical imperfections.

In this paper we have put the analysis of probabilistic metrology on a firm, general footing by using the physically motivated ancilla model.  We generalized prior Fisher-information inequalities to the set of inequalities in Eq.~(\ref{eq:Fisher2}), which apply to all mixed input states, thus including the effect of any technical noise at the input, and to all possible quantum operations for the selection measurement.  This means that the analysis covers all protocols for probabilistic metrology in which MSE is the performance metric; in particular, this encompasses all versions of probabilistic metrology that use weak-value amplification, regardless of whether the defined weak values are real or imaginary.

In the remainder of this section, we formulate three desiderata for analyses of probabilistic quantum metrology.  Then we show how the prior work on using weak-value amplification for metrology and on metrology with abstention is related to the desiderata and to our analysis.

\subsection{Desiderata for probabilistic metrology}\label{subsec:desiderata}

We find it useful to formulate three desiderata for analyses that assess the utility of probabilistic protocols for quantum metrology:

\begin{enumerate}
\item[D1:] Choose a performance metric, and apply it uniformly to all data.
\item[D2:] Include the success probability correctly in the analysis.  This should happen automatically if the problem is set up properly.  Assessing the effect of the success probability might require the sort of statistical analysis given in \srf{sec:results_2}.
\item[D3:] Compare the performance of probabilistic protocols with deterministic protocols, using the same performance metric for all cases.  If possible, compare with the optimal deterministic protocol, which sets a quantum limit on estimation as measured by the chosen performance metric.
\end{enumerate}

\noindent
The analysis we present in Secs.~\ref{sec:parest}--\ref{sec:results_2} adheres to these desiderata by considering single-parameter estimation, with MSE as the performance metric and the corresponding quantum Cram{\'e}r-Rao bound setting the quantum limit on achievable \hbox{MSE}.  The string of inequalities in our Theorem automatically includes the probabilities for favorable outcomes and the overall success probability $p(\checkmark|x)$ in just the right way for comparing probabilistic and deterministic strategies.  In contrast, in much previous work, the relative scarcity of the favorable outcomes and/or the failure to garner any information about the parameters from the discarded outcomes is not fully included in the analysis, thus making it difficult to judge the impact of discarding outcomes.

Though our analysis provides a model for studies of probabilistic quantum metrology, it does not apply directly to much of the previous work for two reasons: some previous work considers multi-parameter estimation, and much of it uses a performance metric other than \hbox{MSE}.  We now take a brief look at some of this previous work to identify problems in the analysis.

\subsection{Weak-value amplification}\label{subsec:weakvalueamp}

Analyses of weak-value amplification~\cite{KofAshNor12,DreMalFil13,JorMarHow13} typically use signal-to-noise ratio (SNR), instead of MSE, as the performance metric and consider a probabilistic protocol successful if the SNR increases on post-selection.  We note that there is no good reason to prefer SNR over MSE, since in the case of amplification, MSE already includes the effects of gain that SNR is meant to capture; moreover, an improvement in SNR does not necessarily imply an improvement in MSE~\cite{FerCom2013}.

Most importantly, as was discussed in~\cite{PanJiaCom13}, the relevant metric for assessing nondeterministic protocols is not the post-selected SNR, which fails to include the success probability $p(\checkmark)$, but rather the root-probability-SNR product, $\sqrt{p(\checkmark)}\times\mbox{SNR}$.  SNR increases as the square root of the number of trials, so when trials proceed to an estimate only on favorable outcomes, the effective number of trials on average is reduced to $p(\checkmark)N$; thus proper accounting requires including $\sqrt{p(\checkmark)}$ in the performance metric.  As in the analysis of nondeterministic immaculate amplifiers in~\cite{PanJiaCom13}, it seems likely that the root-probability-SNR product will show that weak-value amplification does not improve the ability to detect weak signals, in accordance with the related results reported in~\cite{TanYam2013,FerCom2013,Knee2013a,Knee2013b}.

It has been argued that even when the fundamental results presented here and in~\cite{TanYam2013,FerCom2013,Knee2013a,Knee2013b} hold, the situation for weak-value amplification changes when ``technical noise'' is included (see, e.g.,~\cite{JorMarHow13}).  We have yet to see any convincing evidence of this claim because success probability is not properly included in the analysis.  As we noted in \srf{sec:ancilla}, our analysis already includes the effects of technical noise at the input.  Any other technical noise is noise in the measurements and can be regarded as arising from a restriction that prevents the optimal measurement from being performed.  Thus, even were it true that weak-value amplification has advantages in the case of such output technical noise, it would mean that the advantages have nothing to do with fundamental quantum limits and should not be viewed as addressing fundamental questions of quantum mechanics.

\subsection{Metrology with abstention}\label{subsec:abstention}

Protocols for metrology with abstention~\cite{GenRonCal2012,GenRonCal2013a,GenRonCal2013b} have used mean fidelity as the performance metric.  We focus on the protocol considered in~\cite{GenRonCal2012}, which seeks to estimate the direction $\ntrue$ of the Bloch vector of $N$ qubits all of which are in the state $\rho=\frac12(\Id+r\ntrue\cdot\bm{\sigma})$.  The chosen performance metric is the fidelity of the pure qubit state corresponding to the estimate $\nest$ with the pure qubit state corresponding to the true direction $\ntrue$, i.e.,
\begin{equation}
F(\nest,\ntrue)=|\ip{\nest}{\ntrue}|^2=\frac12(1+\nest\cdot\ntrue);
\end{equation}
this fidelity is averaged over the prior distribution of $\ntrue$ and over the measurement results that lead to $\nest$.  If the prior for $\ntrue$ is uniform on the Bloch sphere, the optimal measurement is the covariant measurement, i.e., the measurement that is invariant under simultaneous rotations of the qubits.  This covariant measurement is block diagonal in the angular-momentum subspaces, which have total angular momentum in the range $j=j_{\rm min},\ldots,J$, where $j_{\rm min}=N\,{\rm mod}\,2/2$ and $J=N/2$.  We let $\xi_j$ denote an arbitrary angular-momentum subspace with angular momentum $j$; the number of such subspaces, i.e., the multiplicity of irreducible representations with angular momentum~$j$, is given in Eq.~(11) of~\cite{GenRonCal2012}.  In each angular-momentum subspace, the covariant measurement is a measurement in the basis of angular-momentum coherent states, which are specified by a spin direction; the estimate of spin direction is the result of the measurement.

As shown in~\cite{GenRonCal2012}, the mean fidelity is
\begin{align}
\overline F
&=\int\frac{d\ntrue}{4\pi}\,d\nest\,F(\nest,\ntrue)p(\nest|\ntrue)\nonumber\\
&=\int d\ntrue\,F(\bm{e}_z,\ntrue)p(\bm{e}_z|\ntrue)\nonumber\\
&=\sum_j\int d\ntrue\,F(\bm{e}_z,\ntrue)p(\bm{e}_z|j,\ntrue)p(j|\ntrue).
\end{align}
The second step here follows from the covariance of the measurement, which allows us to pick any direction for the estimate, here $\bm{e}_z$, as long as we integrate over the uniform prior for $\ntrue$.  The same symmetry under rotations implies that $p(j|\ntrue)=p_j$, the probability to find the $N$ qubits with angular momentum~$j$, is independent of $\ntrue$.  The final result for the average fidelity is
\begin{align}
\overline F=\sum_j p_j\overline F_j,
\label{eq:avF}
\end{align}
where
\begin{align}
\overline F_j=\int d\ntrue\,F(\bm{e}_z,\ntrue)p(\bm{e}_z|j,\ntrue)
\end{align}
is the average fidelity for angular momentum~$j$.  Since the fidelity can be thought of as the probability that $\nest$ matches $\ntrue$, the $j$th term in the sum~(\ref{eq:avF}) can be thought of as the probability to get the outcome $j$ times the probability of a match given the outcome~$j$; the overall fidelity is obtained by summing over~$j$.

The abstention protocol regards the identification of total angular momentum~$j$ as a selection measurement; the subsequent identification of a particular subspace within $j$ and the coherent-state measurement in that subspace complete the measurement required to give an estimate.  Since $\overline F_j$ increases with $j$ (because the bigger $j$, the more well-defined the spin direction), the favorable outcomes are chosen to be those whose total angular momentum exceeds a threshold $j_*$, i.e., $\checkmark=\{j_*+1,\ldots,J\}$.  The favorable outcomes have overall probability and average fidelity
\begin{align}
p(\checkmark)=\sum_{j=j_*+1}^J p_j,\qquad
\overline F_\checkmark=\sum_{j=j_*+1}^J\frac{p_j}{p(\checkmark)}\overline F_j;
\end{align}
similarly, the unfavorable outcomes have overall probability and average fidelity
\begin{align}
p(\dingcross)=\sum_{j=j_{\rm min}}^{j_*} p_j,\qquad
\overline F_\dingcrossscript=\sum_{j=j_{\rm min}}^{j_*}\frac{p_j}{p(\dingcross)}\overline F_j.
\end{align}
Since $\overline F_j$ increases with $j$, it is clear that $\overline F_\dingcrossscript<\overline F_\checkmark$.

We can now write a string of inequalities reminiscent of those for quantum Fisher information in \erf{eq:theorem}:
\begin{align}
\overline F=p(\dingcross)\overline F_\dingcrossscript+p(\checkmark)\overline F_\checkmark
\ge\frac12 p(\dingcross)+p(\checkmark)\overline F_\checkmark
\ge p(\checkmark)\overline F_\checkmark.
\label{eq:Fineq}
\end{align}
Given that $\overline F_\dingcrossscript<\overline F_\checkmark$, the first equality says that $\overline F<\overline F_\checkmark$, i.e., that the post-selected averaged fidelity is bigger than the unconditioned average fidelity.  It is this increase in post-selected average fidelity that is reported as the advantage of abstention metrology in~\cite{GenRonCal2012}.  The two inequalities that complete \erf{eq:Fineq}, both of which correspond to discarding information, indicate why this advantage is not useful.  The first inequality says that the average fidelity decreases if one guesses a random spin direction in the event of an unfavorable outcome (random guesses have average fidelity of $1/2$), and the second says that the average fidelity decreases further if one refuses to give an estimate for unfavorable outcomes.  Note that if \erf{eq:Fineq} is divided by $p(\checkmark)$, as is typically done in the literature, the inequalities still hold.

If mean fidelity is the performance metric, the post-selected average fidelity must be multiplied by the probability of a favorable outcome.  To put it a bit differently, in the case of post-selection, the average probability that the estimate matches the true spin direction must include the probability of having the opportunity to make an estimate.

One could repeat the $N$-qubit protocol many times $M$ in the hope that there would be so many favorable outcomes $M_\checkmark$ that $M_\checkmark\overline F_\checkmark>M\overline F$ or, equivalently, that $M_\checkmark>M\overline F/\overline F_\checkmark=\mathbb{E}[M_\checkmark]\overline F/p(\checkmark)\overline F_\checkmark$.  This clearly doesn't work out on average, and the hope can be dashed using the statistical techniques employed in \srf{sec:results_2}, which show that the probability of this happy occurrence decreases exponentially with $M$.

Although the argument we give here is couched in terms of the abstention protocol considered in~\cite{GenRonCal2012}, the same ideas and analysis are easily generalized to any probabilistic protocol that uses fidelity as the performance metric.  The key point is that the probability of favorable outcomes must be included in the post-selected average fidelity.

\section{Discussion}\label{sec:discussion}
Our chief objective in this paper has been to give a rigorous account of quantum limits on probabilistic metrology.  Specifically, we have shown that the quantum Fisher information weighted by the success probability does not increase under post-selection; thus probabilistic metrology cannot improve the quantum limit on the accuracy for estimating a single parameter.  The quantum Fisher information is relevant because we use a quadratic loss function, the MSE of our estimate, as our performance metric; the QCRB tells us that an unbiased estimator can achieve a MSE that is the inverse of the quantum Fisher information.

It remains possible, however, that our conclusion might not hold for other performance metrics applied to other parameter-estimation problems.  We conclude now by briefly considering other possibilities, but caution the reader that to reach a different conclusion about probabilistic metrology requires violating one of the two principles we enunciated in \srf{sec:results_1}. Neither of these principles seems likely to go away.

One approach might be to include a fixed cost for obtaining an unfavorable outcome from the selection measurement.  Additionally (or alternatively), one might use a loss function that penalizes deviations of $\hat x$ from $x_{\rm true}$ more heavily than does a quadratic loss function.  A power-law loss function such as $\En{|\hat x-x_{\rm true}|^n}$ might do that, and the resulting penalty might prejudice one to use states that provide high sensitivity to changes in the parameter.

Though it is conceivable that exotic loss functions or other performance metrics might avoid our negative conclusions about probabilistic quantum metrology, the conjurer of any such function faces three tasks before reporting back to the community.  The first task is provide a detailed account of what parameter-estimation problem the performance metric corresponds to.  The second is to determine, if possible, the ultimate quantum limit---the analogue of the QCRB---on performance in terms of the new metric.  The third is to analyze rigorously the performance of probabilistic protocols as expressed by the new metric, including the effect of success probability in the analysis.  Results without the context that comes from performing these tasks have little call on the attention of those who actually face quantum limitations on measurement precision.

\acknowledgments

This work was supported in part by National Science Foundation Grant Nos.~PHY-1212445 and~PHY-1314763, by Office of Naval Research Grant No.~N00014-11-1-0082.  CF was supported by the Canadian Government through the NSERC PDF program.


\begin{thebibliography}{99}
\bibitem{note}
This follows from $p_1(\partial_x\ln p_1)^2+p_2(\partial_x\ln p_2)^2\ge p[(p_1/p)\partial_x\ln p_1+(p_2/p)\partial_x\ln p_2]^2=p(\partial_x\ln p)^2$, where $p=p_1+p_2$.  The inequality is Jensen's inequality applied to the convex quadratic function.
\end{thebibliography}
\end{document}